\documentclass[10pt,twocolumn,twoside]{IEEEtran}
\usepackage[colorlinks,urlcolor=blue,linkcolor=blue,citecolor=blue]{hyperref}
\usepackage{array}

\usepackage[subrefformat=parens,labelformat=parens,caption=false,font=footnotesize]{subfig}
\usepackage{amsmath}
\usepackage[nolist,nohyperlinks]{acronym}
\usepackage{tikz}
\usetikzlibrary{arrows.meta, positioning}
\usepackage{graphicx}
\usepackage{xcolor}

\begin{acronym}
    \acro{4G}{fourth generation}
    \acro{5G}{fifth generation}
    \acro{BPP}{binomial point process}
\acro{CDF}{cumulative density function}
    \acro{CF}{closed-form}
    \acro{ECDF}{empirical cumulative distribution function}
    \acro{GPS}{global positioning system}
    \acro{GE}{group exploration}
    \acro{GNSS}{global navigation satellite system}
    \acro{LCB}{lower confidence bound}
    \acro{LLR}{log-likelihood ratio}
    \acro{MAB}{multi-armed bandit}
    \acro{MBS}{macro base station}
    \acro{MEC}{mobile-edge computing}
    \acro{mIoT}{massive internet of things}
    \acro{MIMO}{multiple input multiple output}
    \acro{mm-wave}{millimeter wave}
    \acro{mMTC}{massive machine-type communications}
    \acro{MS}{mobile station}
    \acro{MVUE}{minimum-variance unbiased estimator}
    \acro{NLOS}{non line-of-sight}
    \acro{OFDM}{orthogonal frequency division multiplexing}
    \acro{PAC}{probably approximately correct}
    \acro{PDF}{probability density function}
    \acro{PGF}{probability generating functional}
    \acro{PLCP}{Poisson line Cox process}
    \acro{PLT}{Poisson line tessellation}
    \acro{PLP}{Poisson line process}
    \acro{PPP}{Poisson point process}
    \acro{PV}{Poisson-Voronoi}
    \acro{QoS}{quality of service}
    \acro{RAT}{radio access technique}
    \acro{RL}{reinforcement-learning}
    \acro{RSSI}{received signal-strength indicator}
    \acro{BS}{base stations}
    \acro{IoT}{internet of things}
    \acro{IIoT}{industrial internet of things}
    \acro{RF}{radio frequency}
    \acro{WSN}{wireless sensor node}
    \acro{SWIPT}{simultaneous wireless information and power transfer}
    \acro{UE}{user equipment}  
    \acro{wpt}{wireless power transfer}
    \acro{SINR}{signal to interference plus noise ratio}
    \acro{SNR}{signal to noise ratio}
    \acro{SIR}{signal to interference ratio}
    \acro{JSP}{joint success probability}
    \acro{PAoI}{peak age of information}
    \acro{AoI}{age of information}
    \acro{EH}{energy harvesting}
    \acro{UB}{upper bound}
    \acro{LB}{lower bound}
    \acro{AC}{actual case}
    \acro{L}{linear}
    \acro{NL}{non linear}
    \acro{P}{preemptive}
    \acro{NP}{non-preemptive}
    \acro{L EH}{linear energy harvesting}
    \acro{NL EH}{non linear energy harvesting}
\end{acronym}

\usepackage{graphicx,wrapfig,lipsum,subcaption}
\usepackage{rotating}
\usepackage{amsmath, amssymb, amsthm}
\usepackage{stmaryrd}
\usepackage{verbatim}
\usepackage{url}
\usepackage[utf8]{inputenc}
\usepackage[english]{babel}
\newtheorem{theorem}{Theorem}[]

\newtheorem{proposition}{Proposition}[]

\newtheorem{lemma}[]{Lemma}

\usepackage{multicol}

\usepackage{algorithm}
\usepackage{algpseudocode}
\usepackage{scalerel}
\usepackage{setspace}
\newtheorem{definition}{Definition}
\usepackage[noadjust]{cite}
\usepackage{booktabs}
\usepackage{subcaption}
\usepackage{bbm}
\usepackage[normalem]{ulem}
\usepackage{color}
\usepackage{dsfont}
\usepackage{bm}
\usetikzlibrary{automata}
\usetikzlibrary{arrows}
\usetikzlibrary{shapes.geometric, arrows}
\usepackage[]{footmisc}
\usetikzlibrary{positioning}
\usetikzlibrary{calc}
\hypersetup{nolinks=true}
\allowdisplaybreaks
\usepackage{mathrsfs}

\usepackage{caption}
\usepackage{mathtools}
\usetikzlibrary{decorations.pathreplacing}
\usepackage{amsthm}
\newtheorem{remark}{Remark}
\usepackage{tabularx}
\usepackage{amssymb}
\usepackage{acronym} 
\algrenewcommand\algorithmicrequire{\textbf{Input:}}

\begin{document}
\title{Information Age-Controllability Trade-offs in Communication-Constrained Networks}
\author{
    Songita Das, 
    Gourab Ghatak, 
    Chen Quan, 
    and Geethu Joseph 
\thanks{The preliminary version of this paper is submitted to the International Conference on Signal Processing and Communications (SPCOM), 2026. 

S. Das is with the Bharti School of Telecommunication Technology and Management, Indian Institute of Technology Delhi, New Delhi 110016, India (email: bsz228102@iitd.ac.in). G. Ghatak is with the Department of Electrical Engineering, Indian Institute of Technology Delhi, New Delhi 110016, India (email: gghatak@ee.iitd.ac.in). C. Quan and G. Joseph are with the Faculty of Electrical Engineering, Mathematics, and Computer Science, Delft University of Technology, Netherlands (emails: \{c.quan, g.joseph\}@tudelft.nl).}
}
\maketitle
\begin{abstract}
We investigate the trade-off between controllability, channel access, and age-related performance in a wireless network of control systems. Controllers share a random-access channel to transmit control inputs to actuators over slotted blocks. We measure reliable control via block controllability, where a block is controllable if it contains a required number of consecutive successful transmissions. In parallel, we capture information freshness via the age of information. To enable efficient allocation of channel resources over time, we introduce adaptive access probabilities at the block level, prioritizing controllers that have not yet achieved controllability. We then derive closed-form expressions for block controllability probability, the peak latency between inter-block consecutive successes, and peak age of information. We further characterize the peak control latency, defined as the time between consecutive controllable blocks. Finally, we optimize access probabilities to jointly balance controllability and age-related metrics. Numerical results illustrate the effectiveness of the proposed adaptive access policies in managing this trade-off in interference-limited wireless control networks.
\end{abstract}
\begin{IEEEkeywords}
Wireless control-communication co-design, age of information, controllability, stochastic geometry.
\end{IEEEkeywords}
\section{Introduction}
Joint control-communication over wireless networks is central to networked control systems, the internet of things (IoT), and cyber–physical systems, where controllers, sensors, and actuators share a common communication infrastructure. In these systems, control performance is tightly linked to constraints such as packet drops, delays, limited bandwidth, and interference. Hence, channel access and scheduling directly affect the timeliness of state information and control inputs, and should be optimized for control objectives rather than treated as a separate layer. This design of wireless channel access and scheduling jointly with control objectives is the main focus of our paper. 

Beyond classical metrics like throughput and delay, the freshness of information at the controller and actuator is crucial for stabilizing and regulating dynamical systems. The \ac{AoI} captures this freshness by measuring the time elapsed since the latest received update was generated~\cite{Yates,Soysal}. Unlike delay, which concerns individual packets, \ac{AoI} reflects how current the available information is. Therefore, it is especially relevant for real-time control and monitoring~\cite{Kaul} in applications such as industrial automation, intelligent infrastructure, and autonomous transportation~\cite{C1_8272319,C2_552705}. However, control performance and information freshness are jointly shaped by channel access, spatial interference, and transmission reliability, creating fundamental trade-offs. This motivates our framework that jointly analyzes and optimizes communication access and control performance using \ac{AoI}-based metrics \cite{Yates, Abd-Elmagid1} and controllability criteria.
\subsection{Related Work}
We briefly review prior work at the intersection of wireless communication, information freshness, and control \cite{4118465, 9910575} to reveal limitations of existing models and motivate a block-structured, controllability-aware framework. We first discuss studies that incorporate \ac{AoI} into control systems, then review scheduling and access strategies in wireless networks, and finally highlight the gap between slot-based \ac{AoI} optimization and block-level controllability \cite{ghatak2024channelaccessstrategiescontrolcommunication } requirements. 

\ac{AoI} quantifies information freshness in control and monitoring systems~\cite{Yates,Soysal}. Several works analyze its impact on networked control performance. Wang {\it et al.}~\cite{D9305697} derived a relationship between \ac{AoI}, control cost, and retransmission probability in industrial IoT. It showed that low \ac{AoI} improves stability but increases energy consumption due to frequent retransmissions. Hahn {\it et al.}~\cite{2_9599512} incorporated \ac{AoI} prediction into distributed model predictive control, enabling controllers to adapt to anticipated staleness rather than worst-case delays, improving robustness in vehicle platooning. Wen {\it et al.}~\cite{3_10210423} generalized \ac{AoI} to age of task by including sensing, transmission, and computation delays, and proposed event-triggered policies that trade off control cost against end-to-end task freshness. Collectively, these works underscore the role of information freshness but focus on update timing rather than on achieving control objectives. Yet, they do not model the need for consecutive successful transmissions that are common in systems where a sequence of control inputs must be reliably delivered to reach a target state. Consequently, \ac{AoI} evolution is decoupled from probabilistic controllability and the dynamics of repeated transmission attempts. 

Several other studies have explored scheduling and power control mechanisms to improve \ac{AoI} in wireless networks. Qiao {\it et al.}~\cite{8_8772205} optimized transmit power using transmission success/failure feedback to minimize \ac{AoI}. Also, Das {\it et al.}~\cite{J10778324,K11017745} characterized trade-offs between peak \ac{AoI} (PAoI) and joint success probability under slot-based time partitioning in energy-harvesting systems. Learning-based approaches~\cite{G10202236,B6293884,I10778310} have further investigated adaptive scheduling and prediction techniques to improve freshness or reliability. However, these works remain confined to slot-wise or queue-based models and do not account for the fact that controllability may require a run of successive successful transmissions rather than isolated packet deliveries. From a control-centric perspective, Ayan {\it et al.}~\cite{5_10221713} proposed finite-horizon scheduling policies to minimize estimation error under the assumption that control updates can be persistently delivered. In contrast, our work addresses the more fundamental question of when controllability can be achieved and how often it can be maintained over unreliable wireless channels, introducing block controllability and peak control latency as performance metrics. 

To facilitate the analysis, we leverage foundational techniques and analytical tools developed in the communications literature. On the communication side, foundational works by Baccelli {\it et al.}~\cite{A1580787} and Haenggi~\cite{C7345601} analyzed random access protocols and link reliability using stochastic geometry. While these studies provide powerful tools for modeling interference and access contention, they do not consider control-specific requirements such as successive packet delivery, state evolution under packet loss, or information freshness constraints. In short, most communication-control co-design~\cite{E9493202,F9910575} and \ac{AoI} scheduling~\cite{H10354462} studies optimize slot-level reliability or freshness but overlook block-level dynamics required for controllability. It motivates our unified framework that jointly models stochastic access, block controllability, multiple \ac{AoI} metrics, and peak control latency. 
\subsection{Contributions}
We consider a wireless control system where controller locations follow a homogeneous \ac{PPP}. Time is divided into blocks of fixed-length slots. Using both block-wise and slot-wise Aloha access, the model captures the effects of spatial interference, fading, and random channel access on transmission success, information freshness, and control performance. The main contributions are:
\begin{itemize}
\item We introduce a set of block-level performance metrics that jointly capture reliability and freshness in wireless control systems. In particular, we first define block controllability as the event that at least 
$v$ consecutive successful transmissions occur within a block, where $v$ denotes the controllability index of the underlying control system. We further introduce the peak control latency, which measures the number of blocks elapsed between consecutive blocks in which controllability is achieved, providing a block-scale notion of control reliability.
\item We characterize these metrics in Poisson-distributed control networks with block-structured transmissions. Closed-form expressions are derived to compute the probability of block controllability using both block access and slot access schemes via run-length analysis. In addition, we characterize the expected peak control latency by quantifying the number of blocks between consecutive blocks where controllability is achieved. 
\item We investigate the interplay between control performance and information freshness by characterizing the peak latency between inter-block consecutive successes
and \ac{PAoI}. For both paradigms, we derive expressions for the expected latency and \ac{PAoI} under a non-preemptive queue discipline.
\item We propose an optimization framework that dynamically adapts block access, slot access before controllability, and slot access after controllability. We jointly maximize block controllability while minimizing latency and age-related metrics through \ac{CDF}-based cost function. We present a grid-search-based algorithm that computes the optimal per-block access probabilities.
\item  Through extensive numerical results, we demonstrate the evolution of optimal access probabilities with the block index and the controllability requirement $v$. The results
highlight the fundamental trade-offs between block Aloha and slot Aloha access. The proposed adaptive policy favors  block-level access during the pre-controllability phase to improve the likelihood of achieving controllability, while slot-level access in both the pre and  post-controllability regime is more effective in reducing the peak latency between inter-block consecutive successes and the \ac{PAoI}.
\end{itemize}
Compared to our preliminary conference version, which focused primarily on block controllability and peak control latency with adaptive access optimization, this paper substantially extends both the analytical framework and the performance characterization. Specifically, we incorporate information freshness through the peak latency between inter-block consecutive successes and \ac{PAoI}. We further develop a unified communication--control optimization framework that jointly balances controllability, latency, and age-related metrics using adaptive block-level and slot-level access policies together. We also provide regime-specific latency characterizations and significantly expanded numerical results that reveal the evolution of optimal policies across blocks and different controllability indices, offering insights into the trade-offs between controllability, latency, and information freshness.

For ease of reference, the list of notations used throughout the paper is summarized in Table~\ref{tab:key_notations}.
\begin{table}[t]
\renewcommand{\arraystretch}{1.15}
\centering
\small
\begin{tabular}{p{0.23\linewidth} p{0.7\linewidth}}
\hline
\textbf{Symbol} & \textbf{Description} \\
\hline
$\Phi$, $\lambda$ & PPP of controllers and its spatial density \\
$T$, $t$, $j$ & Block length, slot index, controller index; $j=0$ denotes typical controller \\
$r_0$, $r_j$ & Distance from controller to its actuator \\
$i,\,k ,\, k\!-\!\kappa$ & Generic, current, and most recent prior block indices with at least one successful transmission \\
$\kappa$, $v$ & Number of consecutive blocks since last successful block and controllability index (run length)\\
$\mathbf{x}(t)$, $\hat{\mathbf{x}}(t)$, $\mathbf{x}_{\rm des}$ & System state, its estimate and desired state \\
$\mathbf{u}(t)$, $\bar{\mathbf{u}}$ & Control and steady-state control inputs \\
$\mathbf{A}$, $\mathbf{B}$ & System and input matrices \\
$\xi$, $N_o$, $\gamma$, $\Upsilon(t)$ & Transmit power, noise power, SINR threshold and SINR at time $t$ \\
$\mathbb{I}_j(S_t), \mathbb{I}_j(B_k)$ &  Slot and block access indicators of controller $j$ \\
$O_k,  \tilde{O}_k$ & Controllability state and instantaneous block-controllability indicators of controller $j$ \\
$P_{{\rm O}_{k}}, P_{\tilde{\rm O}_k}$ & Controllability state probability and per-block controllability probability \\
$\pi_k$ & First-time block controllability probability \\
$\Phi^{-}$ & Controllers not yet block-controllable \\
$\Phi_{\rm B}, \Phi_{\rm S}, \Phi_{\rm C}$ & Block, slot, and post-controllability access sets \\
$\lambda_{{\rm B}_k}, \lambda_{{\rm S}_k}, \lambda_{{\rm C}_k}$ & Effective densities of controllers in each access mode \\
$\delta_{{\rm B}_k}, \delta_{{\rm S}_k}, \delta_{{\rm C}_k}$ & Access probabilities in block $k$ \\
$G(t)$, $Z(k)$ & Transmission success indicator and the indicator that block $k$ contains at least one successful transmission \\
$L_k,\, l$ & Maximum burst length of consecutive successful transmissions in block $k$ and expected number of successful inputs per successful block \\
$p_{k}, q_k, \varrho_k$ & Per-slot success probability, failure probability ($q_k=1-p_k$), and conditional slot success probability given channel access \\
$\chi(\cdot)$ & Run-length–based controllability function \\
$\Delta^{\rm P}_{\kappa,k}, \mathcal{L}^{\rm P}_{\kappa,k}$ 
& peak age and peak latency of the first successful input in block $k$ \\
$\Theta^{\rm pl}_{k},\, \Theta^{\rm pa}_{k},\, \Theta^{\rm pcl}_{k}$ 
& Peak latency between inter-block consecutive successes, peak age, and peak control latency in block $k$ \\
$P_{\Theta^{\rm curr}_{k}}(\eta)$ 
& CDF of the current-block latency contribution in block $k$, evaluated at threshold $\eta$ \\
$P_{\Theta^{\rm pcl}_{k}}(\eta)$ 
& CDFs of peak control latency in block $k$, evaluated at threshold $\eta$ \\
\hline
\end{tabular}
\caption{Key Notations}
\label{tab:key_notations}
\end{table}
\section{Network Model} 
We consider a shared wireless network that supports multiple controlled systems and their corresponding remote controllers. Each controlled system consists of a sensor, a plant, and an actuator. The network comprises multiple controller–controlled system pairs, where the controller locations are modeled as a homogeneous \ac{PPP}, $\Phi$  with density $\lambda$. The controller–actuator pairs are assumed to be oriented uniformly at random, which, together with the \ac{PPP} of controller locations, form a Poisson dipole process. 

The controllers use a reliable link\footnote{We assume dedicated resources for the sensor–controller link, since this communication occurs sporadically and involves low payloads, whereas the high-payload control channel is shared.} to receive state information at a regular interval of $T$ time slots, referred to as blocks. Then, it calculates control inputs to steer the controlled system to a desired state~\cite{I10778310, 11202619}. These control inputs are transmitted to the actuator over a shared wireless link, and the actuator sends back an acknowledgment. 
The success of each transmission at time $t$ is determined by \ac{SINR} $\Upsilon (t)$,  where acknowledgments $G(t)$ are used by the controller to recompute control inputs if transmission fails. Since a shared wireless communication channel is assumed, not all systems can access the channel simultaneously, as this would degrade performance. Thus, we investigate the channel access policy of each controlled system, which directly influences the SINR and, ultimately, the network's control performance. 

We adopt a probabilistic channel access protocol to manage transmissions and optimize system performance. For tractability, we investigate the system from the perspective of the typical controlled system at the origin. By Slivnyak's theorem, this does not affect the distribution of the remaining network systems~\cite{chiu2013stochastic}. The typical controlled system follows a discrete-time linear dynamical system:
\begin{equation}
    \begin{aligned}
\mathbf{x}(t+1)= \mathbf{Ax}(t)+\mathbf{Bu}(t)+\mathbf{v}(t),
\label{1}
    \end{aligned}
\end{equation}
where $\mathbf{x}(t) \in \mathbb{R}^{n}$ is the state at time slot $t$, $\mathbf{u}(t) \in \mathbb{R}^{m}$  is the control input, $\mathbf{v}(t) \in \mathbb{R}^{n}$ is the zero-mean process noise, and $\mathbf{A} \in \mathbb{R}^{n\times n}$ and $\mathbf{B} \in \mathbb{R}^{n\times m}$ are the state and input matrices of the system. The system drives and retains the system state at the desired state $\mathbf{x}_{\rm des}$ with a suitable choice of inputs $\mathbf{u}(t)$.
\subsection{Propagation Model and Transmission Success}
Each wireless link experiences a Rayleigh-distributed fast fading with parameter 1. The fast fading is independent across time and across all links. Let $\alpha$ be the path-loss exponent of each channel, $N_{\rm o}$ be the channel noise power density, and $\xi$ be the transmit power. The \ac{SINR} $\Upsilon(t)$ at the typical actuator~is
\begin{equation}
    \begin{aligned}
\Upsilon (t)=\frac{\mathbb{I}_{0}(S_{t})\xi  \left | h_{0}(t)\right |^{2} r_{0}^{-\alpha}}{N_{o}+\sum _{j \in \Phi \setminus \{0\}} \mathbb{I}_{j}(S_{t})\xi  \left | h_{j}(t)\right |^{2} r_{j}^{-\alpha}} .
\label{eq:2}
    \end{aligned}
\end{equation}
where $h_{j}(t)$ is the channel fading between controller $j$ and the typical actuator. Here, $j = 0$ is the index of the typical controller and $r_{j}$ denotes the distance between the $j$-th controller and the typical actuator located at the origin. Also, $\mathbb{I}_{j}(S_{t})\in \{0,1\}$ is the channel access state of controller $j$ at time $t$ where $\mathbb{I}_{j}(S_{t})=1$ indicates that controller $j$ transmits and $\mathbb{I}_{j}(S_{t})=0$ indicates no transmission. If $\Upsilon (t)$ exceeds a threshold $\gamma  > 0$ (depends on the application and receiver hardware), the actuator correctly decodes the received signal, thereby implying a successful transmission. So, the transmission acknowledgment sent by the typical actuator to controller $G(t)$ is
\begin{equation}
    \begin{aligned}
G(t)=\begin{cases}
1, &\text{  if  }  \mathbb{I}_{0}(S_{t})\Upsilon (t)>\gamma,\\
0, &\text{ otherwise},
\end{cases}
\label{3}
    \end{aligned}
\end{equation}
where $G(t)=1$ denotes success and $G(t)=0$ denotes failure. In this work, we leverage the statistics of sequence $\{G(t); t \geq 0\}$ to characterize and optimize control, latency and \ac{PAoI}.
\subsection{Control System}
The transmission follows a block-based scheme, where a block comprises $T$ time slots and block $k$ spans the interval
$t \in \{kT, \ldots, (k+1)T-1\}$. The controller periodically receives the system state at the beginning of every block,
$t = kT$. Using the received state $\mathbf{x}(kT)$ and \eqref{1}, it computes an estimate $\hat{\mathbf{x}}(t)$ of the state $\mathbf{x}(t)$ for $kT\leq t\leq(k+1)T-1$ as
\begin{equation}
\begin{aligned}
\hat{\mathbf{x}}(t)
= \mathbf{A}^{t-kT}\mathbf{x}(kT)
+ \sum_{\tau=kT}^{t-1} \mathbf{A}^{t-\tau-1} G(\tau)\mathbf{B}\mathbf{u}(\tau),
\label{4}
\end{aligned}
\end{equation}
as the process noise is zero mean~\cite{I10778310, 11202619}. At $t=kT$, it designs $v$ control inputs ${\mathbf{u}}(t+ \tau), \tau=0,1,\ldots,v-1 $ to drive the predicted system state to the desired target, i.e., $\hat{\mathbf{x}}(t+v) = \mathbf{x}_{\rm des}$. Here $v$ is the controllability index of the system, and the  required control sequence is given by~\cite{I10778310},
\begin{equation}
    \begin{aligned}
\begin{bmatrix}
{\mathbf{u}}(t)^\top\; {\mathbf{u}}(t+1)^\top \ldots {\mathbf{u}}(t+v-1)^\top
\end{bmatrix} ^\top\!\!\!=\boldsymbol{\Psi}^{\dagger}\left [ \mathbf{x}_{\rm des}-\mathbf{A}^{v}\hat{\mathbf{x}}(t) \right ].
\label{6}
    \end{aligned}
\end{equation}
Here, $\boldsymbol{\Psi} = \begin{bmatrix}
\mathbf{ A}^{v-1}\mathbf{B} & \mathbf{A}^{v-2}\mathbf{B} & \dots & \mathbf{B}  
\end{bmatrix}$ is the controllability matrix with $\boldsymbol{\Psi}^{\dagger }$ its pseudo-inverse. 
The controller continues to transmit the control inputs up to time $kT+l$ until either when transmission fails ($G(kT + l) = 0$) or all $v$ inputs are successfully transmitted ($l = v - 1$). If a transmission fails for some $l < v$, the remaining inputs are discarded and a new set of $v$ inputs is recomputed using \eqref{6} with $t = kT + l + 1$ and the process repeats. Once all the $v$ control inputs from \eqref{6} successfully reach the actuator, the state estimate $\hat{\mathbf{x}}(t)$ at the controller reaches $\mathbf{x}_{\rm des}$. 

Once the controller estimate is $\mathbf{x}_{\rm des}$, the controller sends dummy packets till the end of the block.  A dummy packet does not contain control information; it allows the actuator to continue sending acknowledgments, supporting interference monitoring and adaptive channel access. Meanwhile, the actuator applies steady-state inputs till the end of the block to retain the state estimate $\hat{\mathbf{x}}(t)$ at the desired state $\mathbf{x}_{\rm des}$. These steady-state inputs are pre-stored at the actuator. We note that steady-state inputs do not account for process noise; applying them alone can lead to unbounded state deviation from $\mathbf{x}_{\rm des}$ over the full horizon. To compensate for this deviation, the controller recomputes control inputs at the start of each block based on the current state. This periodic correction keeps the state deviation bounded, with its variance growing proportionally to the block length rather than over the entire horizon.

\section{Communication-Control Co-design}
We introduce a block-structured communication framework that prioritizes controllers according to their controllability status to jointly manage reliability and information freshness. To this end, we consider four key metrics: block controllability, peak latency between inter-block consecutive successes, \ac{PAoI}, and peak control latency. The first metric captures instantaneous controllability, while the other three latency-related metrics characterize the temporal evolution of successful communication, information freshness, and controllability.

The controllability is the system's ability to drive its state estimate $\hat{\mathbf{x}}(t)$ to the desired state $\mathbf{x}_{\rm des}$, which occurs only if $v$ consecutive controller transmissions are successful.
\begin{definition}[Controllability]\label{Def:1}
The typical system is said to be {\it block controllable} in block $k$ if there is a run of at least $v$ ones in the sequence $\{G(kT), G(kT +1), \ldots , G((k + 1) T - 1 )\}$.
\end{definition}

While controllability reflects whether control is achieved within a block, it does not characterize the temporal spacing between successful receptions across blocks. To capture this communication latency, we define the peak latency as the time between consecutive inter-block successes. However, to optimize the access probability, we need to associate this metric with a given block $k$. The peak latency is generated only for blocks that contain at least one successful transmission. Blocks without successful transmissions do not create a new peak latency; they only increase the latency until the next successful update. The resulting per-block metric is illustrated in Fig.~\ref{3statement2.drawio}.
\begin{definition}[Peak latency for the first input in a block]\label{def:pl}
For any block $k$ with at least one success, let $\kappa$ denote the number of blocks since the last successful block, and $W_{k-\kappa}$ denote the number of remaining failure slots until the end of that block. Also, let $X_k$ denote the number of failed slots before the first success in block $k$. The peak latency of the first successfully received input in block $k$ is then defined as
\begin{equation}
 \mathcal{L}^{\rm P}_{\kappa,k} = T(\kappa-1)+W_{k-\kappa}+X_k+1.
\end{equation}
\end{definition}

\begin{figure}[t]
\centering   \includegraphics[width=\linewidth]{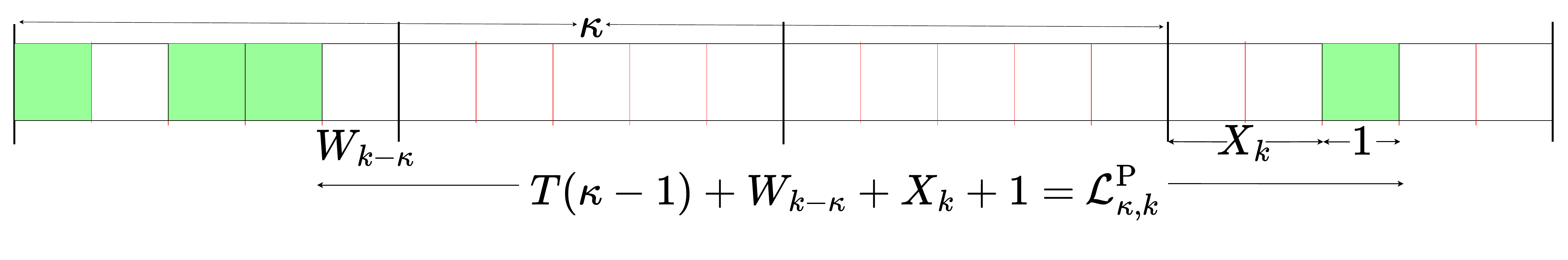}
 \caption{Sample \ac{AoI} evolution under peak latency between inter-block consecutive successes with block length $  T=5  $, where successful transmissions are shown in green. After the last success in block $k-\kappa$, $W_{k-\kappa}$ failure slots occur, followed by $\kappa-1$ failed blocks contributing $T(\kappa-1)$ slots. In block $k$, the first success occurs after $X_k$ failures at slot $X_k+1$, giving latency $\mathcal{L}^{\rm P}_{\kappa,k} = T(\kappa-1) + W_{k-\kappa} + X_k + 1$.}
\label{3statement2.drawio}
\end{figure}   

The next latency metric, \ac{PAoI}, captures the information freshness at the actuator. We recall that the controller receives a fresh measurement from the sensor only at the beginning of each block $k$ (at time $t=kT$). All control inputs generated and transmitted during block $k$ are therefore computed based on the system state sensed at time \(kT\). \ac{PAoI} measures the maximum information staleness immediately before a successful input reception, i.e., the maximum elapsed time between the generation time of the previously received successful input and the reception time of the current input. 
\begin{definition}[\ac{PAoI}]
Let $\iota=1,2,\dots$ index successfully received control inputs in the order of their reception at the actuator. If input $\iota$ is successfully received at slot $t$, the \ac{PAoI} associated with that control input is 
\begin{equation}
\Delta^{\rm P}_{\iota}
=
t - t_{\iota-1},
\end{equation}
where $t_{\iota(t-1)}$ is the generation time of the most recently received successful packet prior to the current reception.
\end{definition}
To express PAoI in block form, note that it is generated only in blocks with at least one success; failed blocks only increase the \ac{AoI} to the next update. The resulting per-block metric is illustrated in Fig.~\ref{state3}.
\begin{definition}[PAoI of the first input in a block] \label{def:paoi}
For any block $k$ with at least one success ($Z(k)=1$), let $\kappa$ denote the number of blocks since the last successful block, and let $X_k$ denote the number of failed slots before the first success in block $k$. The PAoI of the first successfully received input in block $k$ is defined as
\begin{equation}
\Delta^{\rm P}_{\kappa,k} = \kappa T + (X_k + 1).
\end{equation}
\end{definition}

\begin{figure}[t]
\centering   \includegraphics[width=\linewidth]{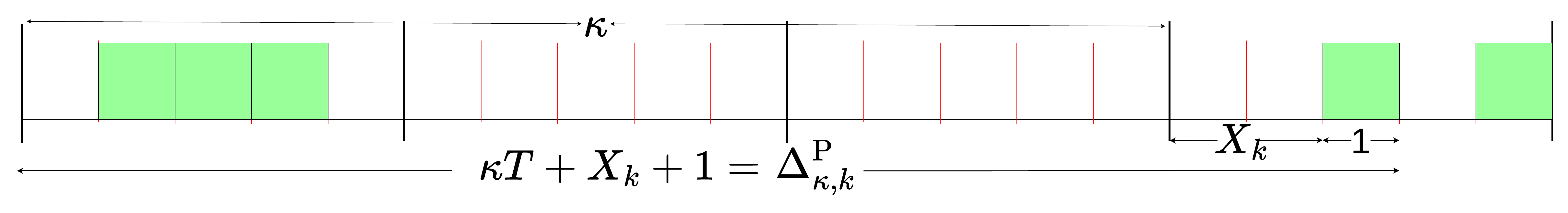}
\caption{Sample \ac{AoI} evolution in non-consecutive blocks with block length $T=5$, where successful transmissions are shown in green. In the $k$th block, the first successful transmission occurs after an initial failure run of length $X_{k}$, so $\Delta^{\rm P}_{\kappa,k}=\kappa T+X_{k}+1$, where $\kappa$ denotes the number of consecutive failed blocks preceding the $k$th block. }
\label{state3}
\end{figure}

\begin{figure*}[t]
\centering   \includegraphics[width=1\textwidth]{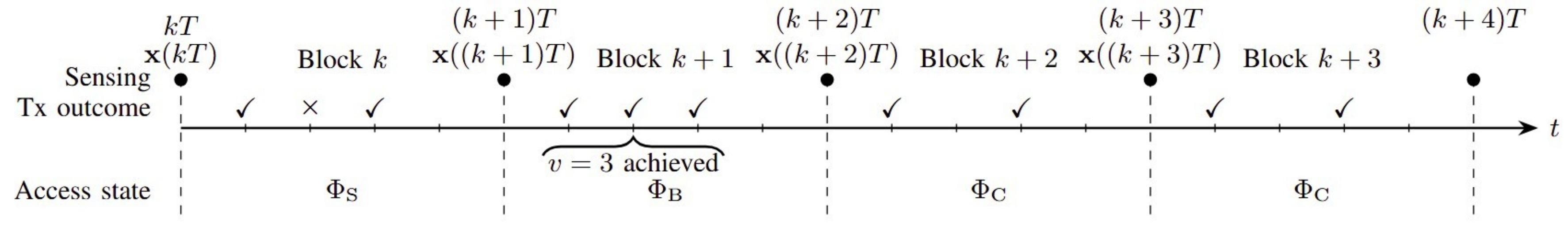}
\caption{Timing illustration for $T=5$ and $v=3$ showing block boundaries, state sensing, slot-level transmission outcomes,  formation of a burst of $v$ consecutive successful transmissions, and the evolution of access strategies from slot access ($\Phi_{\rm S}$) to block access ($\Phi_{\rm B}$), and finally to post-controllability slot access, where, {$\Phi_{\rm B},\Phi_{\rm S}\rightarrow\Phi_{\rm C}$}.}
\label{fig:twc_system_events}
\end{figure*}
Finally, to capture long-term control reliability, we introduce the peak control latency, which measures the number of blocks elapsed since the last controllable block and reflects the temporal consistency of maintaining controllability. To track controllability, we define the indicator variable $\tilde{O}_k=1$ if the typical controller achieves block controllability in block $k$.
\begin{definition}[Peak control latency]\label{def:pcontrolage}
For a controllable block $k$, the peak control latency $\Theta^{\rm pcl}_k \triangleq \tau\in\{1,\dots,k\}$ where $k-\tau$ denotes the most recent block prior to $k$ in which controllability was achieved\footnote{For consistency in our notations, we define that $\tilde{O}_0=1$.}. 
\end{definition}
We note that $\Theta^{\rm pcl}_k$ is defined only for blocks with $\tilde{O}_k=1$ and represents the gap between controllable blocks. 

Having defined the metrics, we next introduce controller scheduling. We consider a service horizon of $K$ blocks, each consisting of $T$ time slots. At the beginning of block $k$, the controller set $\Phi$ is partitioned into two disjoint subsets. The set $\Phi^{-}$ contains controllers that have not achieved block controllability up to the end of block $k-1$, while the set $\Phi_{\rm C} \triangleq \Phi \backslash \Phi^{-}$ consists of those controllers that have. By definition, before the initial block, no controller has yet achieved controllability, and hence $\Phi^{-} = \Phi$. 

To improve the controllability probability for the controllers in $\Phi^-$, we adopt a block Aloha protocol with access probability $\delta_{{\rm B}_k}$. This partitions $\Phi^{-}$ into $\Phi_{\rm B}$ with the controllers that gain block access and $\Phi_{\rm S}$ with the controllers that do not. Controllers in $\Phi_{\rm B}$ access the channel in all slots of block $k$, whereas the remaining controllers in $\Phi_{\rm S}$ and $\Phi_{\rm C}$ access the channel according to a slotted Aloha protocol with access probabilities 
$\delta_{\rm{S}_{k}}$ and $\delta_{\rm{C}_{k}}$, respectively. The access probabilities $\delta_{{\rm B}_k}, \delta_{{\rm S}_k}$ and $\delta_{{\rm C}_k}$ directly impact the slot transmission success probability, and thus controllability, latency, and \ac{PAoI}.  This design prioritizes controllers that have not yet achieved controllability. The probabilistic block access controls the density of transmitters that are active in all slots, reducing interference while enabling adaptive scheduling.

We optimize these access probabilities across blocks to balance the controllability, latency, and \ac{AoI} by integrating controllability probability, \ac{AoI} evolution, and peak control latency. Specifically, we analyze block controllability in Section~\ref{sec:alohasucess} and characterize the latency-related metrics in Section~\ref{sec:paoi}. Finally, these metrics are jointly optimized to determine the access probabilities in Section~\ref{sec:optimization}.
\subsection{Controllability Analysis}\label{sec:alohasucess}
To prioritize controllers that have not yet achieved controllability, we characterize and optimize the probability of achieving controllability. Specifically, we derive the block controllability probability, 
$P_{{\rm O}_{k}}
\triangleq
\mathbb{P}(O_k = 1),$ where the indicator variable $O_k=1$ denotes at least one block within $\{1,2,\ldots,k\}$ has been controllable, i.e., $\tilde{O}_i=1$ for some $1\leq i \le k$. 

Although the analysis is performed for a typical controller, by Slivnyak’s theorem, it is representative of all controllers in the PPP $\Phi$. Hence, the access probabilities $(\delta_{\rm{B}_k}, \delta_{\rm{S}_k}, \delta_{\rm{C}_k})$ define a network-wide randomized policy, under which each controller belongs to one of three access modes $\Phi_{\rm B}$, $\Phi_{\rm S}$, or $\Phi_{\rm C}$ with probabilities $(1-P_{{\rm O}_{k-1}})\delta_{\rm{B}_k}$, $(1-P_{{\rm O}_{k-1}})(1-\delta_{\rm{B}_k})$, and $P_{{\rm O}_{k-1}}$, respectively, implying that these probabilities correspond to the fractions of controllers in each set. The controllers not yet block-controllable form a fraction $1-P_{{\rm O}_{k-1}}$, while the remaining $P_{{\rm O}_{k-1}}$ are already in the post-controllability regime. Accordingly, the spatial densities $\lambda_{{\rm B}_k}$, $\lambda_{{\rm S}_k}$, and $\lambda_{{\rm C}_k}$ of potentially transmitting controllers in the sets $\Phi_{\rm B}$, $\Phi_{\rm S}$, and $\Phi_{\rm C}$ at block $k$ are respectively, 
\begin{align}
\lambda_{{\rm B}_k} &= (1-P_{{\rm O}_{k-1}})\,\delta_{{\rm B}_k}\lambda &
\lambda_{{\rm S}_k} &\!=\! (1-P_{{\rm O}_{k-1}})(1-\delta_{{\rm B}_k})\delta_{{\rm S}_k}\lambda \notag\\ \lambda_{{\rm C}_k} &= P_{{\rm O}_{k-1}}\,\delta_{{\rm C}_k}\lambda. \label{lambdas}
\end{align}
Here, we start with  $\lambda_{{\rm B}_1} =\delta_{{\rm B}_1}\lambda$, $\lambda_{{\rm S}_1} = (1-\delta_{{\rm B}_1})\delta_{{\rm S}_1}\lambda$, and $\lambda_{{\rm C}_1} = 0$, as no controller has achieved block controllability before block $k=1$. 
 Using these densities, the slot-success probability is derived as follows.
\begin{lemma}\label{lemma:1}
The conditional slot transmission success probability $\varrho_{k} =\mathbb{P}\!\left(\Upsilon(t)>\gamma \middle| S_t=1, \phi\right)$ of the typical controller in block $k$ is independent of the access mode $\phi$ and is given by
\begin{equation*}
    \begin{aligned}
\varrho_{k} \!=\exp\Big(-\frac{\gamma N_{o}}{\xi  r_{0}^{-\alpha}}\Big)   \exp \Big({\!-2 \pi \lambda_{\rm eff}\!\! \int_{0}^{\infty }\!\!\!\!\frac{\gamma r^{-\alpha}}{r_{0}^{-\alpha}\!+\! \gamma r^{-\alpha}}rdr}\! \Big), \!\!\!\!
    \end{aligned}
\end{equation*}
where the parameters $\xi,N_o$ and $\alpha$ are defined in~\eqref{eq:2}, and $\gamma$ is the SINR threshold for successful transmission, and $\lambda_{\rm eff}=\lambda_{\rm{B}_{k}}+\lambda_{\rm{S}_{k}}+\lambda_{\rm{C}_{k}}$ is the transmitting controller density from~\eqref{lambdas}. 
\end{lemma}
\begin{proof}
See Appendix \ref{AppendixA}.  
\end{proof}
\begin{remark}
An exact analysis conditioned on a specific realization $  \Phi  $ of the PPP is analytically possible using higher-order moment analysis and meta distributions, as performed in \cite{ghatak2024channelaccessstrategiescontrolcommunication}. However, in practical uncoordinated wireless networked control systems, individual controllers do not have access to the exact locations of interferers. The only available information is typically the density $  \lambda  $. The mean-field approach therefore provides a robust success probability by averaging over all possible spatial geometries, favorable cases (distant interferers), unfavorable cases (nearby strong interferers), and all intermediate configurations, for the given density. The resulting deterministic $  \varrho_k  $ depends only on the effective density $  \lambda_{\rm eff}  $, enabling tractable optimization of the access probabilities $  \delta_{\rm B_k}  $, $  \delta_{\rm S_k}  $, and $  \delta_{\rm C_k}  $ that performs well across both good and poor interference geometries.
\end{remark}
The next lemma connects the slot-success and controllability probabilities.
\begin{lemma}\label{chiLemma}
In block $k$ of $T$ slots, the controllability probability $\mathbb{P}(\tilde{O}_k=1)=\chi(\mathbb{P}(\Upsilon(t)>\gamma))$ for any slot $t$ in block $k$, where  
\begin{align*}
\!\chi(x)\!=\!\!\!\sum_{l=1}^{\left \lfloor \!\frac{T+1}{v+1}\!\right \rfloor}\!\!\!(\!-\!1)^{l+1}\!\Bigg[ x\!+\!\frac{T\!-\!lv\!+\!1}{l} (1\!-\!x) \!\Bigg]\!\!\binom{T\!-\!lv\!}{l-1}x^{lv}(\!1\!-\!x)^{l-1}\!,
\end{align*}
with controllability index $v$ and  SINR threshold $\gamma$.
\end{lemma}
\begin{proof}
See Appendix \ref{Appendix:Chi}.  
\end{proof}
Combining the above two lemmas, we can compute the block controllability probability  $P_{{\rm O}_{k}}$ recursively as follows.
\begin{theorem}\label{lem:Chi}
Consider a control network following a PPP with density $\lambda$ and access probabilities $\delta_{{\rm B}_k}, \delta_{{\rm S}_k}$ and $\delta_{{\rm C}_k}$. The controllability state probability up to block $k$ is 
\begin{equation} \label{eq:state}
P_{{\rm O}_{k}}
=
P_{{\rm O}_{k-1}}
+
(1-P_{{\rm O}_{k-1}})\pi_k,
\quad k > 1, \text{ with } P_{\rm{O}_1}=\pi_1,
\end{equation}
where $\pi_k$ is the first-time controllability probability at $k$,
\begin{equation}\label{eq:blockprobability}
\pi_{k}\!=\!\mathbb{P}(\tilde O_k \!=\! 1 \!\mid\! O_{k-1}\!=\!0)\!=\!
\delta_{{\rm B}_k}\chi(\varrho_{k})\!+\!
(1\!-\!\delta_{{\rm B}_k}) \chi(\delta_{{\rm S}_k} \varrho_{k}),
\end{equation}
where the function $\chi(\cdot)$ is defined in Lemma~\ref{chiLemma}.
\end{theorem}
\begin{proof}
The result follows because $O_k=1$ if $O_{k-1}=1$, or if controllability is achieved for the first time in block $k$. Since these two events are mutually exclusive, we derive \eqref{eq:state}. 

Further, to derive \eqref{eq:blockprobability}, we note that conditioned on $O_{k-1}=0$, the controller belongs to $\phi=\Phi_{\rm B}$
with probability $\delta_{{\rm B}_k}$ or $\phi=\Phi_{\rm S}$ with probability $1-\delta_{{\rm B}_k}$. 
Hence, 
\begin{multline*}
    \pi_{k}=  \delta_{{\rm B}_k}\mathbb{P}(\tilde{O}_k=1|O_{k-1}=0,\phi=\Phi_{\rm B})\\+
(1-\delta_{{\rm B}_k}) \mathbb{P}(\tilde{O}_k=1|O_{k-1}=0,\phi=\Phi_{\rm S}).
\end{multline*}
Now, invoking Lemmas~\ref{lemma:1} and ~\ref{chiLemma}, we arrive at \eqref{eq:blockprobability}, completing the proof.
\end{proof}
\subsection{Latency Analysis}\label{sec:paoi}
We next characterize three latency-related metrics: the peak latency and \ac{PAoI} of the first input in a block and peak control latency, as defined in Definitions~\ref{def:pl}, \ref{def:paoi}, and \ref{def:pcontrolage}.

We recall that at the beginning of every block \(k\), we select the optimal global access probabilities \(\delta_{\rm{B}_k}\), \(\delta_{\rm{S}_k}\), and \(\delta_{\rm{C}_k}\) based on the fraction of controllers that have achieved controllability up to block $k-1$. Also, we do not track the specific regime ($\Phi_{\rm B}$, $\Phi_{\rm S}$, or $\Phi_{\rm C}$) of individual controllers; hence, past information is limited to whether each block was successful or not, i.e., \(Z(i)=1\) for $i<k$, where \(Z(i)\) is the indicator variable denoting that block \(i\) contains at least one successful transmission. 

For a given per-block success probability $p_i$ for the $i$th block, the slot success indicators within block $i$ are i.i.d.\ Bernoulli random variables due to independent fading. However, the slot success probability \(p_i\) in block \(i\) is itself a random variable, because it depends on the unknown regime of the typical controller in that block and evolves as the access strategy changes over time, as illustrated in Fig.~\ref{fig:twc_system_events}. Based on this observation, we derive the expected peak latency for the first input in block $k$ with $Z(k)=1$, using the past slot success probabilities.

\begin{theorem}\label{Theo:1}
For a given block $k$ such that $Z(k)=1$, given the history of success probabilities \(\{p_1,\dots,p_k\}\),  the expected peak latency for the first input is given by
\begin{align}
\Theta^{\rm pl}_{k}
&=
\sum_{\kappa=1}^{k}
(1-q_{k-\kappa}^T)
\left[
\prod_{i=k-\kappa+1}^{k-1}q_i^T
\right]
\left[
\frac{q_{k-\kappa}}{p_{k-\kappa}}
+
T\kappa
\right]
\notag\\
&\quad
-
T
\!\sum_{\kappa=1}^{k}\!
\left[
\prod_{i=k-\kappa}^{k-1}\!\!q_i^T
\right]
\!-\!T\!
+\!
\bigg( \frac{q_k}{p_k}
\!-\!
\frac{Tq_k^T}{1-q_k^T} \!\bigg)
\!+\!1,
\label{Theorem1}
\end{align}
where we define $q_i=1-p_i$.
\end{theorem}
\begin{proof}
See Appendix \ref{AppendixB}.
\end{proof}

In a similar spirit, we derive the expected \ac{PAoI} of the first control input in block $k$ with $Z(k)=1$, using the past slot success probabilities. Unlike the peak latency metric, \ac{PAoI} focuses on information freshness and depends only on the elapsed time since the previous successful update.
\begin{theorem}[PAoI]
\label{Theo:2}
For a given block $k$ such that $Z(k)=1$, given the history of success probabilities \(\{p_1,\dots,p_k\}\),  the expected \ac{PAoI} of the first control input is given by 
\begin{align}
\Theta^{\rm pa}_{k}
\!=\!
T\!
\sum_{\kappa=1}^{k}\!
\kappa
(1\!-\!q_{k-\kappa}^T) \!\!\!
\prod_{i=k-\kappa+1}^{k-1}\!\!\!\!q_i^T
\!+\!
\bigg( \frac{q_k}{p_k}
\!-\!
\frac{Tq_k^T}{1-q_k^T} \!\bigg)
\!+ \!
1. \label{th2}
\end{align}
\end{theorem} 
\begin{proof}
See Appendix \ref{AppendixC}.
\end{proof}

Finally, the next result characterizes the peak control latency of a controllable block.   
\begin{theorem}\label{thm:controllatency}
Consider a control network following a PPP with density $\lambda$ and access probabilities $\delta_{{\rm B}_k}, \delta_{{\rm S}_k}$ and $\delta_{{\rm C}_k}$. Using $\varrho_k$ and $\chi(\cdot)$  in Lemmas~\ref{lemma:1} and~\ref{chiLemma}, the peak control latency satisfies
\begin{equation} 
\mathbb{P}(\Theta^{\rm pcl}_k\!=\!\tau| \tilde{O}_k=1)
\!=\!
\frac{
P_{\tilde{\rm{O}}_{k-\tau}}
\prod_{i=k-\tau+1}^{k-1}\big(1-\chi(\delta_{{\rm C}_i}\varrho_i)\big)
}{
\sum_{\kappa=1}^{k}
\!\!P_{\tilde{\rm{O}}_{k-\kappa}}
\!\!\prod_{i=k-\kappa+1}^{k-1}\!\!\big(1\!-\!\chi(\delta_{{\rm C}_i}\varrho_i)\big)
},
\label{eq:CA_pmf_exact}
\end{equation}
for $\tau=1,2,\ldots,k$. Here, instantaneous controllability probability $P_{\tilde{\rm{O}}_k}$ is defined using $P_{{\rm O}_{k-1}}$ and $\pi_k$ in Theorem~\ref{lem:Chi}~as
\begin{equation} \label{eq:inst}
P_{\tilde{\rm{O}}_k}=\mathbb{P}(\tilde{O}_k=1)
=
(1-P_{{\rm O}_{k-1}})\pi_k
+
P_{{\rm O}_{k-1}}\chi(\delta_{{\rm C}_k}\varrho_k).
\end{equation}
\end{theorem}
\begin{proof}
See Appendix~\ref{AppendixD}.
 \end{proof}
 Using Theorem~\ref{thm:controllatency}, the expected peak control latency at a controllable block $k$ is 
\begin{align}
\mathbb{E}[\Theta^{\rm pcl}_k \!\mid \!\tilde{O}_k\!=\!1]
\!=\!
\frac{
\sum_{\tau=1}^{k}
\tau
P_{\tilde{\rm{O}}_{k-\tau}}
\!\prod_{i=k-\tau+1}^{k-1}(1\!-\!\chi(\delta_{{\rm C}_i}\varrho_i))
}{
\sum_{\kappa=1}^{k}
P_{\tilde{\rm{O}}_{k-\kappa}}
\prod_{i=k-\kappa+1}^{k-1}(1-\chi(\delta_{{\rm C}_i}\varrho_i))
}.
\label{eq:E_CA_exact}
\end{align}

Building on the controllability and latency analysis, we next formulate the optimization framework for access probabilities.

\subsection{Access Probability Optimization}\label{sec:optimization}
We optimize access probabilities  \(\delta_{\rm B_k}\), \(\delta_{\rm S_k}\), and \(\delta_{\rm C_k}\) by balancing two key objectives: minimizing the latency metrics while maximizing the block controllability probability. However, from the peak latency and \ac{PAoI} analysis, we note that the first two terms in~\eqref{Theorem1} and the first term in \eqref{th2} depend only on the history \(\{p_1,\dots,p_{k-1}\}\) and are independent of the current access probability. The only term relevant to the optimization is the $\left(
\frac{q_{k}}{p_{k}}
- \frac{T q_{k}^T}{1-q_{k}^T}
\right)$. So, we minimize the expectation of this term over the regime uncertainty,
\begin{equation}\label{eq:Thetacurr}
\Theta^{\rm curr}_k=\!\!\!\!\!\!\!\!\!\!\sum_{\phi\in\{\Phi_{\rm B},\Phi_{\rm S},\Phi_{\rm C}\}}
\!\!\!\!\!\!\!\!\!\mathbb{P}(\phi \mid Z(k)=1)
\left(
\frac{q_{k,\phi}}{p_{k,\phi}}
- \frac{T q_{k,\phi}^T}{1-q_{k,\phi}^T}
\right).
\end{equation}
Here, $q_{k,\phi}=1-p_{k,\phi}$ and the regime-specific slot success probabilities are 
\begin{equation}\label{eq:p_kphi}
    p_{k,\phi}=\begin{cases}
        \varrho_k &\text{for } \phi = \Phi_{\rm B}\\
        \delta_{\rm S_k}\varrho_k &\text{for } \phi = \Phi_{\rm S}\\
        \delta_{\rm C_k}\varrho_k &\text{for } \phi = \Phi_{\rm C}.
    \end{cases}
\end{equation}
Also, the regime probabilities conditioned on block success~are
\begin{align}
 \mathbb{P}(\phi\mid Z(k)=1)&= \frac{\mathbb{P}(Z(k)=1\mid\phi)P_{k,\phi}}{\sum_{\phi'\in\{\Phi_{\rm B},\Phi_{\rm S},\Phi_{\rm C}\}}\mathbb{P}(Z(k)=1\mid\phi')P_{k,\phi'}} \notag\\
 &= \frac{
P_{k,\phi}\bigl[1-q_{k,\phi}^T\bigr]
}{
\!\sum_{\phi' \in \{\Phi_{\rm{B}}, \Phi_{\rm{S}}, \Phi_{\rm{C}}\}}\!\!
P_{k,\phi'}\!\bigl[1\!-\!q_{k,\phi'}^T\bigr]
},\label{rp}
\end{align}
where the regime fractions are 
\begin{equation}
    P_{k,\phi} = \begin{cases}
        (1-P_{{\rm O}_{k-1}})\delta_{\rm{B}_k} &\text{for } \phi = \Phi_{\rm B}\\
        (1-P_{{\rm O}_{k-1}})(1-\delta_{\rm{B}_k})\delta_{\rm{S}_k} &\text{for } \phi = \Phi_{\rm S}\\
        P_{{\rm O}_{k-1}}\delta_{\rm{C}_k}&\text{for } \phi = \Phi_{\rm C},
    \end{cases}\label{rff}
\end{equation}
since the probability that the typical controller belongs to set $\phi$ is equal to the fraction of controllers in that same set $\phi$. The formulation using $\Theta^{\rm curr}_k$ isolates the impact of current block decisions, avoids incorporating terms that depend on the past, and eliminates the need to account for all possible system trajectories.

Furthermore, to make both the latency terms and controllability comparable on a probabilistic scale, we express the latency objective via its cumulative distribution function. We define the CDF of \(\Theta^{\rm curr}_k\) as
\begin{align}\label{eq:jointtheta_curr}
P_{\Theta^{\rm curr}_k}(\eta)
&=
\mathbb{P}(\Theta^{\rm curr}_k \le \eta, Z(k)=1) \notag
\\&=
\mathbb{P}(\Theta^{\rm curr}_k \le \eta \mid Z(k)=1)\, \mathbb{P}(Z(k)=1).
\end{align}

Similarly, we characterize peak control latency by its conditional CDF conditioned on block $k$ being controllable. The metric is defined as
$\mathbb{P}(\Theta^{\rm pcl}_{k}\le\eta \mid \tilde{O}_k=1),
\quad \eta \ge 0$ and its corresponding joint form $\mathbb{P}(\Theta^{\rm pcl}_k \le \eta, \tilde{O}_k=1)$. Using the definition of conditional probability, we obtain
\begin{align}   
P_{\Theta^{\rm pcl}_{k}}(\eta)
=\! \mathbb{P}(\Theta^{\rm pcl}_k \!\le\! \eta, \tilde{O}_k\!=\!1) \!= \mathbb{P}(\Theta^{\rm pcl}_k \le \eta \mid \tilde{O}_k=1)\, P_{\tilde{\rm{O}}_k}.
\label{control_age}
\end{align}
Using the CDF-based age characterization, we formulate cost functions that jointly balance block controllability, latency,  input age, and peak control latency to optimize access.

We define a policy in block $k$ as the triplet
$(\delta_{\rm{B}_k}, \delta_{\rm{S}_k}, \delta_{\rm{C}_k})$ that balances controllability, latency and age related performance. A candidate policy is any feasible triplet in the set $[0,1]^3$ considered during the numerical search. We jointly maximize the block controllability probability $P_{{\rm O}_k}$, the likelihood that the current-block latency contribution satisfies $\Theta^{\rm curr}_k \leq \eta$, and the likelihood that the peak control latency satisfies $\Theta^{\rm pcl}_k \leq \eta$. Accordingly, the cost function for block $k$ is defined as
\begin{align}
J(\delta_{\rm{B}_k},\delta_{\rm{S}_k},\delta_{\rm{C}_k})
=
P_{{\rm O}_k}
+
\rho_1 P_{\Theta^{\rm curr}_k}(\eta)
+
\rho_2 P_{\Theta^{\rm pcl}_k}(\eta),
\label{cost_new}
\end{align}
where $\rho_1,\rho_2 \in (0,1)$ are weighting parameters that control the relative importance of controllability and latency-related performance. The optimal policy for block $k$ is obtained by 
\begin{align}\label{eq:opt}    
\!\!\!\!(\delta_{\rm{B}_k}^*,\delta_{\rm{S}_k}^*,\delta_{\rm{C}_k}^*)
=
\arg\max_{\delta_{\rm{B}_k},\delta_{\rm{S}_k},\delta_{\rm{C}_k}\in[0,1]}
J(\delta_{\rm{B}_k},\delta_{\rm{S}_k},\delta_{\rm{C}_k}).
\end{align}
Since the resulting optimization problem is nonconvex and admits no closed-form solution, the optimal access probabilities are obtained numerically via an exhaustive grid search over the feasible parameter space. For each block $k$, all candidate policies are evaluated in terms of controllability, slot success probability, current-block latency contribution, and peak control latency. The CDF terms in the cost function are evaluated empirically over the set of candidate policies within the grid, thereby ranking policies according to their relative performance. The policy that maximizes the cost is then selected for that block, as summarized in Algorithm~\ref{alg:candidate_eval}. 

As a final remark, although intuitively it seems that maximizing per-slot success probability \(p_k\) would also maximize the block controllability probability, the following proposition shows that this is not always true.
\begin{proposition}\label{proposition1}
Maximizing the per-slot transmission success probability within a block is not equivalent to maximizing the probability of achieving a run of \(v\) consecutive successful transmissions within that block.
\end{proposition}
\begin{proof}
See Appendix~\ref{AppendixE}.
\end{proof}
Although the controllability metric and latency-related performance may align under certain parameter settings (see Section~\ref{sec:numerical}), this behavior is not universal. In particular, the current-block latency contribution \(\Theta^{\rm curr}_k\) and the controllability probability depend on different underlying mechanisms, namely the slot success probability \(p_k\) and the run-length success probability \(\chi(p_k)\), respectively. As a result, in other regimes, such as larger values of \(v\), higher interference density \(\lambda\), or different block lengths \(T\), the optimal policies for these objectives can diverge. Therefore, joint optimization of controllability and latency via the unified CDF-based cost function remains essential.

\begin{algorithm}[t]
\caption{Access Policy Optimization for Block \(k\)}
\label{alg:candidate_eval}
\begin{algorithmic}[1]
\Require Block index \(k\), \(P_{{\rm O}_{k-1}}\), histories $\{p_1,\dots,p_{k-1}\}$, $\{P_{\tilde O_1},\dots,P_{\tilde O_{k-1}}\}$, $\{\chi_{C_1},\dots,\chi_{C_{k-1}}\}$
\Ensure Optimal access policy \((\delta_{\rm{B}_k}^*,\delta_{\rm{S}_k}^*,\delta_{\rm{C}_k}^*)\) and all performance metrics for block \(k\)

\For{each candidate triplet \((\delta_{\rm{B}_k},\delta_{\rm{S}_k},\delta_{\rm{C}_k}) \in [0,1]^3\)}


    \State Compute effective interference density \(\lambda_{\rm eff}\) using \eqref{lambdas}

    \State Compute conditional success probability \(\varrho_k\) using Lemma~\ref{lemma:1}

    \State Compute regime-specific slot-success probabilities $p_{k,\phi}$ using \eqref{eq:p_kphi}


    \State Compute state-transition probability \(\pi_k\) using \eqref{eq:blockprobability}, \(P_{{\rm O}_{k}}\) and \(P_{\tilde{\rm{O}}_k}\) using \eqref{eq:state} and \eqref{eq:inst}


    \State Compute \(\Theta^{\rm curr}_k\) using \eqref{eq:Thetacurr}-\eqref{rff}

    \State Compute peak control latency using \eqref{eq:CA_pmf_exact}-\eqref{eq:E_CA_exact}

    \State Compute joint probabilities 
    \(P_{\Theta^{\rm curr}_k}(\eta)\) and 
    \(P_{\Theta^{\rm pcl}_k}(\eta)\) using \eqref{eq:jointtheta_curr} and \eqref{control_age} and cost $J$

\EndFor

\State Select the candidate that maximizes \(J\)

\State Return the optimal policy and corresponding metrics

\end{algorithmic}
\end{algorithm} 
\section{Numerical Results}\label{sec:numerical}
We evaluate the proposed framework to characterize the trade-offs between controllability and latency under the current-block optimization model. Unless otherwise specified, the system parameters are fixed as \(T=5\), \(\lambda=10^{-4}\,\mathrm{m}^{-2}\), \(\alpha=3\), \(\gamma=0.1\), transmit power \(\xi=40\) dBm, noise power \(N_o=10^{-17}\) and latency threshold \(\eta=3\). The optimal access probabilities \((\delta_{\rm{B}_k}, \delta_{\rm{S}_k}, \delta_{\rm{C}_k})\) are obtained via exhaustive grid search over \([0,1]^3\) by maximizing the cost function. 
\begin{figure*}[hpt]
	\centering
        \subfloat[]{
        \!\!\!\!\!\includegraphics[trim=1cm 0cm 0cm 0.8cm,width=0.34\textwidth, height=4.2cm]{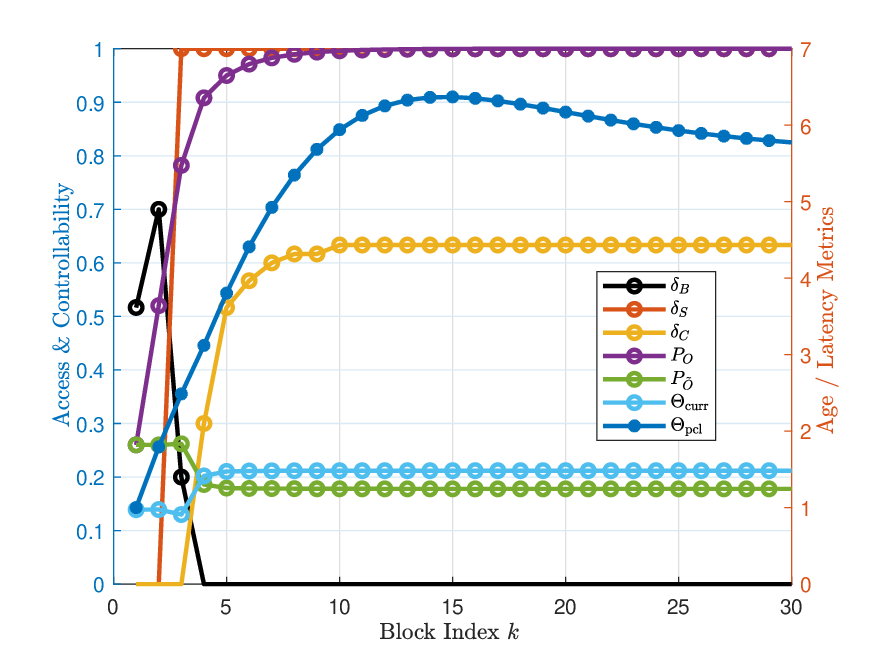}
        }
        \subfloat[]{
   		\!\!\!\!  \includegraphics[trim=0.8cm 0cm 0.2cm 0.8cm,width=0.34\textwidth, height=4.2cm]{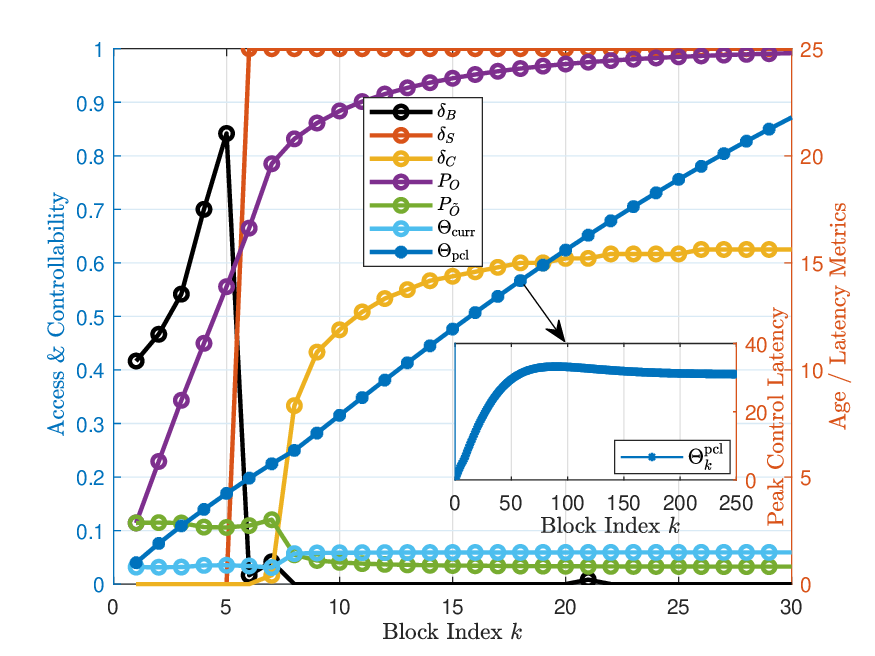}
        }
        \subfloat[]{
   		\includegraphics[trim=0.8cm 0cm 1cm 0.6cm,width=0.28\textwidth, height=4.2cm]{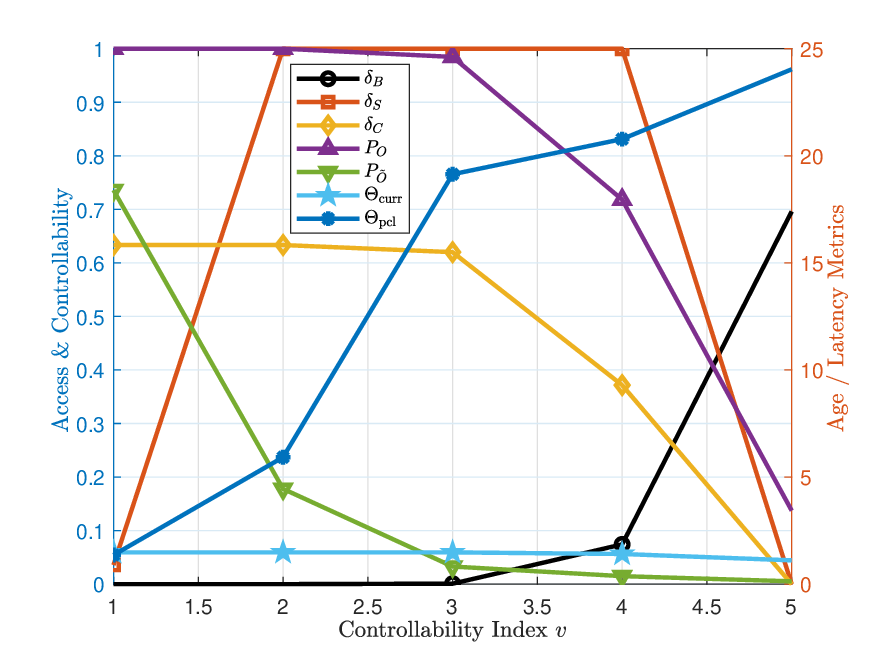}
        }
	\caption{Evolution of optimal access probabilities, controllability, and age/latency metrics over (a) blocks for $v=2$, (b) blocks for $v=3$, and (c) controllability index $v$.}
	\label{fig_abc}
\end{figure*}
\begin{figure}[hpt]
	\centering
   		\!\!\!\!\!\!\!\!\!\! \!\includegraphics[trim=0cm 0.4cm 0cm 0.8cm, height=4.2cm]{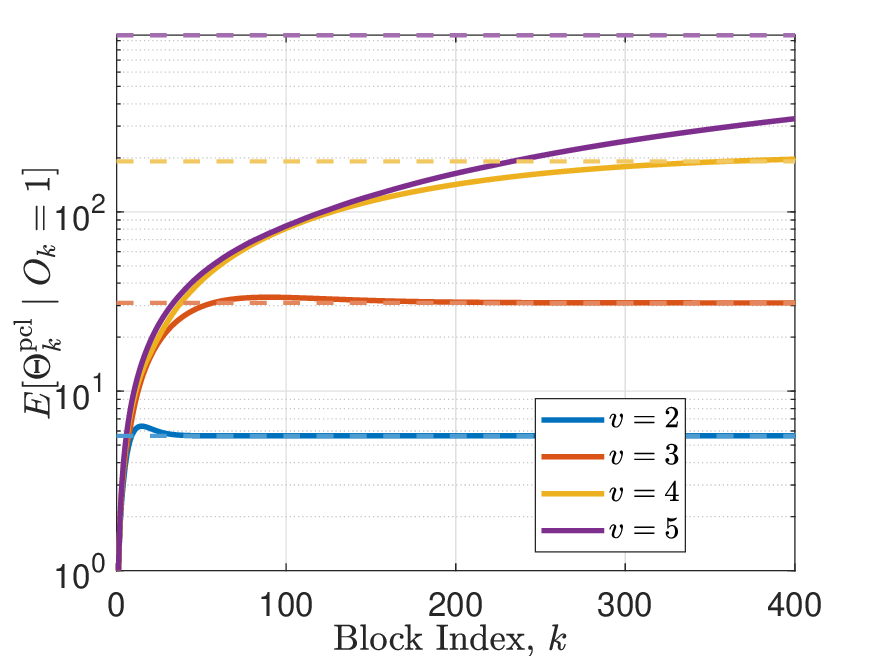}
	\caption{Expected peak control latency (log-scale y-axis) versus block index for $v=2,3,4,5$, showing higher latency for larger $v$.}
	\label{fig_ab}
\end{figure}

Figs.~\ref{fig_abc}(a) and~\ref{fig_abc}(b) illustrate the evolution of the optimal access probabilities, controllability metrics, and latency measures over blocks for controllability indices \(v=2\) and \(v=3\), respectively. The system exhibits two distinct regimes, a pre-controllability phase, where \(P_{{\rm O}_{k}}<1\), and a post-controllability phase, where \(P_{{\rm O}_{k}}\to 1\). In the pre-controllability phase, the optimizer prioritizes maximizing the state-transition probability \(\pi_k\), which drives the system toward controllability. As a result, the block-access probability \(\delta_{\rm{B}_k}\) increases rapidly in the initial blocks. Once the controllability probability approaches unity, i.e., \(P_{{\rm O}_{k}}\approx 1\), the system transitions to the post-controllability regime. In this regime, \(\delta_{\rm{B}_k}\) becomes inactive and drops to zero, while \(\delta_{\rm{S}_k}\) saturates to one. The system dynamics are then governed solely by the post-controllability access probability \(\delta_{\rm{C}_k}\), which balances interference and controllability. The controllability probability \(P_{{\rm O}_{k}}\) increases significantly faster for smaller values of \(v\), since achieving a run of \(v\) consecutive successes becomes more difficult as \(v\) increases. Consequently, for \(v=2\), controllability is achieved within a few blocks, whereas for \(v=3\), the transition is delayed. The instantaneous controllability probability \(P_{\tilde{\rm{O}}_k} = \chi(\delta_{\rm{C}_k}\varrho_k)\) attains a higher steady-state value for smaller \(v\), reflecting the reduced difficulty in achieving shorter success runs.
The current-block latency contribution \(\Theta^{\rm curr}_k\) stabilizes rapidly in both cases, since it depends only on the per slot success probabilities within a block. In contrast, the peak control latency \(\Theta^{\rm pcl}_k\) exhibits significantly different behavior. During the initial phase, \(\Theta^{\rm pcl}_k\) increases due to the accumulation of blocks that have not yet achieved controllability. As the system approaches the post-controllability regime, \(\Theta^{\rm pcl}_k\) converges to a steady-state value determined by the steady-state controllability probability.

Fig.~\ref{fig_abc}(c) illustrates the steady-state behavior of the system as a function of the controllability index \(v\). As \(v\) increases, the instantaneous controllability probability \(P_{\tilde{\rm{O}}_k}\) decreases rapidly, reflecting the increasing difficulty of achieving longer runs of consecutive successful transmissions. Correspondingly, the peak control latency \(\Theta^{\rm pcl}_k\) increases significantly with \(v\), indicating a higher delay in achieving controllability. From the access policy perspective, the optimizer adapts by increasing \(\delta_{\rm{B}_k}\) at higher values of \(v\), thereby encouraging more aggressive transmissions in the pre-controllability phase to improve the likelihood of achieving the required success runs. At the same time, \(\delta_{\rm{C}_k}\) decreases with \(v\), reducing interference once the system approaches or reaches controllability. The access probability \(\delta_{\rm{S}_k}\) remains high for moderate values of \(v\) but drops sharply at larger \(v\), indicating a shift in strategy as the system prioritizes controllability over maintaining moderate access. In contrast, the current-block latency \(\Theta^{\rm curr}_k\) remains nearly constant across all values of \(v\), showing that instantaneous latency is relatively insensitive to the controllability requirement. Overall, the figure highlights that increasing \(v\) primarily impacts controllability and long-term control delay, while having minimal effect on the current-block latency.

Fig.~\ref{fig_ab} illustrates the evolution of the conditional peak control latency \(\mathbb{E}[\Theta^{\rm pcl}_k \mid \tilde{O}_k=1]\) over \(K=400\) blocks for \(v\in\{2,3,4,5\}\). For \(v=2\), the latency quickly stabilizes at a low value, indicating that controllability is easily maintained. For \(v=3\), the latency increases initially and then saturates at a moderate value. For larger values of \(v\), such as \(v=4\) and \(v=5\), the latency grows significantly and may exhibit near-linear growth over a long horizon before reaching steady state. This behavior reflects the increasing difficulty of achieving long consecutive success runs in the presence of interference. In steady state, when \(P_{{\rm O}_{k}}\to 1\), all controllers operate in the post-controllability regime with a fixed access probability \(\delta_C^*\) and constant success probability \(\varrho\). Then, the instantaneous controllability probability becomes \(P_{\tilde{\rm{O}}_k} = \chi(\delta_C^*\varrho)\), which is independent of \(k\). Therefore, we obtain
\(
\lim_{k\to\infty}\mathbb{E}[\Theta^{\rm pcl}_k \mid \tilde{O}_k=1]
= \frac{1}{\chi(\delta_C^*\varrho)}
\) \cite{9360520}.
This theoretical asymptotic value is represented by the horizontal dashed lines in Fig.~\ref{fig_ab}. As \(k\) becomes large, the curves converge perfectly to these lines, confirming that the numerical results match the derived steady-state limit.
\section{Conclusion}
This paper investigated block-structured wireless control systems where multiple controller–actuator pairs share a random access channel using adaptive Aloha-based transmission. We introduced three adaptive access probabilities for block access, pre-controllability slot access, and post-controllability slot access. Analytical expressions were derived for block controllability probability, peak latency between inter-block successes, peak age of information, and peak control latency. Using a CDF-based optimization framework, we jointly balanced controllability and freshness-related metrics by selecting optimal access probabilities. Numerical results demonstrated the effectiveness of the proposed policies and highlighted the trade-off between reliable control and timely information transmission in interference-limited wireless control networks.
\appendices
\section{Proof of Lemma~\ref{lemma:1}}
\label{AppendixA}
For any $\phi$, the set of interfering controllers is $\bar{\Phi}_t=\{j\in\Phi_{\rm B}\cup\Phi_{\rm S}\cup\Phi_{\rm C}\setminus\{0\}: \mathbb{I}_{j}(S_{t})=1\}$. Here, the indicator $\mathbb{I}_{j}(S_{t})=1$ if controller $j$ transmits in slot $t$. The process $\bar{\Phi}_t$ follows a PPP with density $\lambda_{\rm eff}$.  
Now, using \eqref{eq:2}, 
\begin{equation*}
\begin{aligned}
\varrho_{k}&= \mathbb{P}\!\left( \!\frac{\xi  |h_{0}(t)|^{2}r_{0}^{-\alpha} }{
N_{o} +\sum_{j\in \bar{\Phi}_{t}}\xi  |h_{j}(t)|^{2}r_{j}^{-\alpha}} >  \gamma  \!\right)\\
&= \mathbb{P}\left(  |h_{0}(t)|^{2}>
\frac{\gamma N_{o}}{\xi  r_{0}^{-\alpha}} + \sum_{j\in \bar{\Phi}_{t}}|h_{j}(t)|^{2}\gamma\frac{r_j^{-\alpha}}{r_{0}^{-\alpha}}  \right)\\
&= \exp\left(-\frac{\gamma N_{o}}{\xi  r_{0}^{-\alpha}}\right)\mathbb{E}\left[ e^{- \sum_{j\in \bar{\Phi}_{t}}|h_{j}(t)|^{2}\gamma\frac{r_j^{-\alpha}}{r_{0}^{-\alpha}} }\right]\\
\end{aligned}
\end{equation*}
because $|h_0(t)|^2\!\sim\!\exp(1)$. Further, since $|h_{j}(t)|^{2}\gamma\frac{r_j^{-\alpha}}{r_{0}^{-\alpha}}\sim\exp\left(\frac{r_0^{-\alpha}}{\gamma r_{j}^{-\alpha}}\right)$, and $\bar{\Phi}_t$ follows a PPP with density $\lambda_{\rm eff}$, we simplify the above expression using probability generating functions~\cite{9013286} to obtain the desired result.
\section{Proof of Lemma~\ref{chiLemma}}  
\label{Appendix:Chi}
It suffices to show that $\chi(x)$ is the probability that a Bernoulli sequence $\{Y(t)\in\{0,1\}\}_{t=1}^T$ with mean $x$ contains a run of at least $v$ ones. 
For $i'=1,2,\dots,T-v+1$, we define the event $\mathcal{A}_{i'} \triangleq \{Y(i')=Y(i'+1)=\cdots=Y(i'+v-1)=1\}$, corresponding to a run of $v$ consecutive successes starting at slot~$i'$. Then, $\chi(x)=\mathbb{P}\Big(\bigcup_{i'=1}^{T-v+1} \mathcal{A}_{i'}\Big).$ Simply summing the probabilities $\mathbb{P}(A_{i'})$ overcounts outcomes with multiple runs. By the principle of inclusion-exclusion,
\begin{equation}
\chi(x)
=
\sum_{l=1}^{T-v+1}
(-1)^{l+1}
\!\!\!\!\!
\sum_{1\le i'_1<\cdots<i'_l\le T-v+1}
\!\!\!\!\!
\mathbb{P}(A_{i'_1}\cap\cdots\cap A_{i'_l}).
\label{eq:IE}
\end{equation}
The first term $(l=1)$ adds the probabilities of individual runs. The second term $(l=2)$ subtracts the probabilities of pairs of runs that occur together, correcting for double counting. Higher-order terms alternately add and subtract probabilities of three or more simultaneous runs. To ensure $l$ runs of length $v$ are distinct and non-overlapping, at least one empty slot must separate consecutive runs. Thus, $(l-1)$ gaps are needed among $lv$ occupied slots, so $lv + (l-1) \le T$, which gives $l \le \left\lfloor \frac{T+1}{v+1} \right\rfloor$. For any admissible selection of $l$ nonoverlapping runs, the joint probability that all $l$ runs occur is
\begin{equation}\label{eq:PA_prob}
\mathbb{P}(A_{i'_1}\cap\cdots\cap A_{i'_l})
=
x^{lv}(1-x)^{l-1},
\end{equation}
where $x^{lv}$ accounts for $lv$ successful slots within the runs, and
$(1-x)^{l-1}$ ensures separation between consecutive runs.
After allocating $lv$ slots to the runs, remaining $T-lv$ slots can be
distributed among the $(l-1)$ separating gaps, yielding
$\binom{T-lv}{l-1}$ admissible placements. Runs starting at the first slot or ending at the last slot create boundary effects because they need a separating failure on only one side, whereas interior runs require failures on both sides. In line with classical run distributions (e.g., de Moivre’s run problem), this yields the correction factor $x+\frac{T-lv+1}{l}(1-x)$, where $X_{k}$ accounts for boundary-attached runs (those starting at the first slot or ending at the last slot). The term $\frac{T-lv+1}{l}(1-x)$ accounts for interior runs: $T-lv+1$ is the number of admissible starting positions, division by $l$ averages over the $l$ runs, and $(1-x)$ enforces a separating failure. Substituting this into \eqref{eq:IE} and \eqref{eq:PA_prob} yields Lemma~\ref{chiLemma}.
\section{Proof of Theorem~\ref{Theo:1}}  
\label{AppendixB}
For a given block $k$ with success, let $k-\kappa$ for $\kappa \geq 1$ denote the prior block with at least a successful transmission. Then,  
the latency can be decomposed as
$\mathcal{L}^{\rm P}_{\kappa,k}
=
W_{k-\kappa}
+
T(\kappa-1)
+
X_k
+
1.$
Here, $X_{k}$ is the length of the first run of failures at the start of the block, and $W_{k-\kappa}$ is the length of the last run of failures at the end of the block (see Fig.~\ref{3statement2.drawio}).

Conditioned on $Z(k- \kappa)=1$ and $Z(k)=1$, we have  
\begin{multline*}
\mathbb{E}[\mathcal{L}^{\rm{P}}_{\kappa,k}|Z(k- \kappa)=1, Z(k)=1] 
=\mathbb{E}[W_{k-\kappa} \mid Z(k-\kappa)=1]  \\
+
T(\kappa-1)
+
\mathbb{E}[X_k \mid Z(k)=1]
+
1.
\end{multline*}
Applying the law of total expectation over $\kappa$, we derive
\begin{align}
\Theta^{\rm pl}_{k}&= \mathbb{E}\left[\mathcal{L}^{\rm P}_{\kappa,k} \mid Z(k)=1\right] \notag\\&= 
\sum_{\kappa=1}^{k} 
\mathbb{P}(\kappa \mid Z(k)=1) \mathbb{E}[\mathcal{L}^{\rm P}_{\kappa,k} \mid Z(k-\kappa)=1,    Z(k)=1]\notag\\
&= \sum_{\kappa=1}^{k} 
\mathbb{E}\left[W_{k-\kappa} \mid Z(k-\kappa)=1\right]\mathbb{P}\left(\kappa \mid Z(k)=1\right)\notag\\
&\quad T\mathbb{E}[\kappa \mid Z(k)=1]-T + \mathbb{E}\left[X_k \mid Z(k)=1\right]
+
1,
\label{eq:latency_total}
\end{align}
where $\mathbb{P}\!\left(\kappa \!\mid \!Z(k)\!=\!1\right)=\mathbb{P}\!\left(Z(k-\kappa)=1 \!\mid \!Z(k)\!=\!1\right)$.
The rest of the proof computes probability and expectations in~\eqref{eq:latency_total}. 
\subsubsection{Expected slots after last success within a block, $W_{k-\kappa}$}
Since slot outcomes are i.i.d. Bernoulli with success probability $p_{k-\kappa}$ and failure probability $q_{k-\kappa}$, the last success is followed by $n_w$ failures if the final $n_w$ slots are failures and the preceding slot is a success. Hence, for $0\le n_w\le T-1$,
we have $P(W_{k-\kappa}=n_w \mid Z(k-\kappa)=1)
=
\frac{p_{k-\kappa}\, q_{k-\kappa}^{\,n_w}}
{1-q_{k-\kappa}^T}.$
Therefore, we deduce
\begin{align}
\mathbb{E}[W_{k-\kappa} \mid Z(k-\kappa)=1]&=
\sum_{n_w=0}^{T-1}
n_w
\frac{p_{k-\kappa}\, q_{k-\kappa}^{\,n_w}}
{1-q_{k-\kappa}^T}\notag\\
&=
\frac{
q_{k-\kappa}
-
T q_{k-\kappa}^T
+
(T-1)q_{k-\kappa}^{T+1}
}{
p_{k-\kappa}
\left[1-q_{k-\kappa}^T\right]
}\notag\\
& = \frac{q_{k-\kappa}}{p_{k-\kappa}}-
\frac{T q_{k-\kappa}^{T}
}{1-q_{k-\kappa}^T
}
.\label{eq:exp_1}
\end{align}
\subsubsection{Expected number of blocks between two consecutive successful blocks}
Let $\kappa$ be the number of blocks since the most recent success before block $k$, so the last success occurs in block $k-\kappa$. Since the system resets after each success and the channel is i.i.d., the inter-success epochs form an i.i.d. renewal process. For $\kappa \ge 1$, if the last success is in block $k-\kappa$, the next in block $k$, and all intermediate blocks fail, i.e., $Z(k-\kappa+1)=\dots=Z(k-1)=0$, then the joint probability that blocks $k-\kappa$ and $k$ each contain at least one success is 
\begin{multline*}
    \mathbb{P}(Z(k-\kappa)=1,Z(k)=1)
\\= (1-q_{k-\kappa}^T)(1-q_k^T)\Bigg[\prod_{i=k-\kappa+1}^{k-1}q_i^T\Bigg].
\end{multline*}
Thus, the conditional probability that the previous successful block is $k-\kappa$ given that block $k$ is successful is
\begin{align}
\mathbb{P}(\kappa \mid Z(k)=1)
&=
\frac{\mathbb{P}(Z(k-\kappa)=1, Z(k)=1)}
{\mathbb{P}(Z(k)=1)}    \notag\\&
=
\bigl[1-q_{k-\kappa}^T\bigr]
\prod_{i=k-\kappa+1}^{k-1}q_i^T,   
\label{eq:prob1}
\end{align}
where $\kappa \ge 1$. The case $\kappa = k$ corresponds to the situation where no successful block occurred prior to block $k$. This case is well-defined since we define $\tilde{O}_0=1$, which implies that block $0$ acts as a virtual successful block. Consequently, the support of $\kappa$ is $\{1,2,\dots,k\}$. The conditional expectation of $\kappa$ given $Z(k)=1$ is
\begin{align}
    \mathbb{E}[\kappa|Z(k)=1]&=\sum_{\kappa=1}^{k}\!\kappa  \Bigg[(1\!-\!q_{k- \kappa}^T) \prod_{i=k- \kappa+1}^{k-1}q_i^T\Bigg].
\label{eq:kappa:}
\end{align}
\subsubsection{Expected length of the initial failure run within a block}
Since the block contains at least one success, $X_{k}$ can take values in $\{0,1,\dots,T-1\}$. If the first $n_x$ slots are failures and the $(n_x+1)$-th slot is a success,
then $$\mathbb{P}(X_{k}=n_x) = q_k^{\,n_x} p_k,$$ for $n_x=0,1,\dots,T-1,$
where $p_k$ and $q_k=1-p_k$ denote the slot success and failure probabilities, respectively. 
Conditioning on $Z(k)=1$ yields  $$\mathbb{P}(X_{k}=n_x \mid Z(k)=1)
=
\frac{q_k^{\,n_x} p_k}{1-q_k^T},
\quad n_x=0,1,\dots,T-1. $$  
Then, similar to \eqref{eq:exp_1}, the conditional expectation of $X_{k}$ simplifies to
\begin{align}
\mathbb{E}[X_{k}\mid Z(k)=1]
&= \sum_{n_x=0}^{T-1} n_x \frac{q_k^{\,n_x} p_k}{1-q_k^T}
=\frac{
q_k}{
p_k}-\frac{
Tq_k^{T}
}{
1-q_k^T
}. 
\label{eqX}
\end{align}
\subsubsection{Expected peak latency for the first input in the block}
Substituting \eqref{eq:exp_1} and \eqref{eq:prob1} into the first term of \eqref{eq:latency_total}, we obtain
\begin{align}
&\sum_{\kappa=1}^{k}
\mathbb{E}\!\left[W_{k-\kappa}\mid Z(k-\kappa)=1\right]
\mathbb{P}\!\left(\kappa \mid Z(k)=1\right)
\notag\\
&
=\!
\sum_{\kappa=1}^{k}\!
\left[
\prod_{i=k-\kappa+1}^{k-1}\!\!\!\!\!q_i^T
\right]
\frac{q_{k-\kappa}}{p_{k-\kappa}}
(1\!-\!q_{k-\kappa}^T)
\!-\!
T
\sum_{\kappa=1}^{k}\!
\left[
\prod_{i=k-\kappa}^{k-1}\!\!\!q_i^T
\right].
\label{eq:step2}
\end{align}
Substituting ~\eqref{eq:kappa:}-\eqref{eq:step2} into \eqref{eq:latency_total}, we obtain
\begin{align*}   
\Theta^{\rm pl}_{k}
&\!=\!
\sum_{\kappa=1}^{k}\!
\left[
\prod_{i=k-\kappa+1}^{k-1}\!\!\!q_i^T
\right]
\!\frac{q_{k-\kappa}}{p_{k-\kappa}}
(1\!-\!q_{k-\kappa}^T)
\!-\!
T
\sum_{\kappa=1}^{k}\!
\left[
\prod_{i=k-\kappa}^{k-1}\!\!q_i^T
\right]
\notag\\
&\quad
+\!
T\!
\sum_{\kappa=1}^{k}\!
\kappa
(1\!-\!q_{k-\kappa}^T)
\left[
\prod_{i=k-\kappa+1}^{k-1}\!\!\!\!q_i^T
\right]
\!-\!T
\!+\!
\frac{q_k}{p_k}
\!-\!
\frac{Tq_k^T}{1-q_k^T}
\!+\!1\\
&=
\sum_{\kappa=1}^{k}
(1-q_{k-\kappa}^T)
\left[
\prod_{i=k-\kappa+1}^{k-1}q_i^T
\right]
\left[
\frac{q_{k-\kappa}}{p_{k-\kappa}}
+
T\kappa
\right]
\notag\\
&\quad
-
T
\sum_{\kappa=1}^{k}
\left[
\prod_{i=k-\kappa}^{k-1}q_i^T
\right]
\!-\!T
+ \bigg(
\frac{q_k}{p_k}
\!-\!
\frac{Tq_k^T}{1-q_k^T} \!\bigg)
\!+\!1.  
\end{align*}
Hence, we arrive at \eqref{Theorem1} and the proof is complete.
\section{Proof of Theorem~\ref{Theo:2}} 
\label{AppendixC}

The proof is similar to that of Theorem~\ref{Theo:1} in Appendix~\ref{AppendixB}.
Conditioned on $\kappa$ and $Z(k)=1$, we have $$\mathbb{E}\!\left[\Delta^{\rm P}_{\kappa,k} \mid \kappa, Z(k)=1\right]
=
\kappa T
+
\mathbb{E}[X_{k}\mid Z(k)=1]
+
1.$$
Applying the law of total expectation over $\kappa$, we obtain 
\begin{align*}
\Theta^{\rm pa}_{k}  &= \mathbb{E}\left[\Delta^{\rm P}_{\kappa,k} \mid Z(k)=1\right]
\\&=  
\sum_{\kappa=1}^{k}
\mathbb{P}\left(\kappa \mid Z(k)=1\right)  
\left(
\kappa T
+
\mathbb{E}[X_{k}\mid Z(k)=1]
+
1
\right)\\
&=T\,\mathbb{E}[\kappa \mid Z(k)=1]
+
\mathbb{E}[X_{k}\mid Z(k)=1]
+
1.
\end{align*}
Substituting $\mathbb{E}[\kappa \!\mid \!Z(k)\!=\!1]$ from~\eqref{eq:kappa:} and 
$\mathbb{E}[X_{k}\!\mid \!Z(k)\!=\!1]$ from ~\eqref{eqX}, we obtain the closed-form expression for \(\Theta^{\rm pa}_k\).
\begin{align}
\!\!\!\!\!\Theta^{\rm pa}_{k}
\!=\!
T
\sum_{\kappa=1}^{k} \!
\kappa
(1\!-\!q_{k-\kappa}^T) \!\!\!\!
\prod_{i=k-\kappa+1}^{k-1}\!\!\!\!q_i^T
\!+\!
\bigg( \frac{q_k}{p_k}
\!-\!
\frac{Tq_k^T}{1-q_k^T} \bigg) \!
+ \!
1.
\label{eq:pa_final}
\end{align}
Hence, the theorem is proved.
\section{Proof of Theorem~\ref{thm:controllatency}} 
\label{AppendixD}
From Definition~\ref{def:pcontrolage}, for any controllable block $k$, we get
 \begin{align*}
     \mathbb{P}(\Theta^{\rm pcl}_k\!=\!\tau| \tilde{O}_k=1 ) \\
&\hspace{-3cm} = \mathbb{P}(
\tilde{O}_{k-\tau}\!=\!1,
\tilde{O}_{k-\tau+1}\!=\!\dots\!=\!
\tilde{O}_{k-1}\!=\!0,
\tilde{O}_k\!=\!1)/\mathbb{P}(\tilde{O}_k\!=\!1)\\
&\hspace{-3cm}=P_{\tilde{\rm{O}}_{k-\tau}}\mathbb{P}(
\tilde{O}_{k-\tau+1}\!=\!\dots\!=\!
\tilde{O}_{k-1}\!=\!0,\tilde{O}_k\!=\!1|
\tilde{O}_{k-\tau}\!=\!1
)/P_{\tilde{\rm{O}}_k}.
\end{align*}
If $\tilde{O}_{k-\tau}=1$, all subsequent blocks $i>k-\tau$ are in the post-controllability regime with slot access $\delta_{{\rm C}_i}$, implying $\mathbb{P}(\tilde{O}_{i}=1|\tilde{O}_{k-\tau}=1,\tilde{O}_{k-\tau+1},\dots,\tilde{O}_{i-1}) =\chi(\delta_{{\rm C}_i}\varrho_i) $. Therefore, using the chain rule, we deduce 
\begin{equation*} 
\mathbb{P}(\Theta^{\rm pcl}_k\!=\!\tau| \tilde{O}_k=1)
\!=\!\frac{P_{\tilde{\rm{O}}_{k-\tau}}}{P_{\tilde{\rm{O}}_k}}
\!\!\prod_{i=k-\tau+1}^{k-1}\big(1-\chi(\delta_{{\rm C}_i}\varrho_i)\big)\chi(\delta_{{\rm C}_k}\varrho_k).
\end{equation*}
Computing the denominator $P_{\tilde{\rm{O}}_k}$ from the law of total probability by
conditioning on the previous controllable state $k-\kappa$, we derive \eqref{eq:CA_pmf_exact}. Then, Theorem~\ref{lem:Chi} implies
\begin{align*} 
P_{\tilde{\rm{O}}_k}
&= \sum_{o=0,1}\mathbb{P}(O_{k-1}=o)\mathbb{P}(\tilde{O}_k=1|O_{k-1}=o)\\
&=(1-P_{{\rm O}_{k-1}})\pi_k
+
P_{{\rm O}_{k-1}}\chi(\delta_{{\rm C}_k}\varrho_k),
\end{align*}
as $\mathbb{P}(\tilde{O}_k=1 \mid O_{k-1}=1) = \chi(\delta_{{\rm C}_k}\varrho_k)$, deriving \eqref{eq:inst}.

\section{Proof of Proposition~\ref{proposition1}} 
\label{AppendixE}

We prove the claim by counterexample, focusing on the first block \(k=1\), where $P_{{\rm O}_0}=0,$
so only the pre-controllability access probabilities \((\delta_{{\rm B}_1},\delta_{{\rm S}_1})\) are active.
Consider the special case \(T=v=2\). In this case, block controllability requires two consecutive successful transmissions, and so $\chi(x)=x^2.$

Let \(d = \delta_{{\rm B}_1} + (1-\delta_{{\rm B}_1})\delta_{{\rm S}_1}\) denote the effective access probability. From the expressions for the spatial densities in~\eqref{lambdas} (with \(P_{{\rm O}_0}=0\) and \(\lambda_{{\rm C}_1}=0\)), the total effective density of transmitting controllers is
\begin{align*}\lambda_{\rm eff} 
&= \lambda_{{\rm B}_1} + \lambda_{{\rm S}_1} 
= \delta_{{\rm B}_1}\lambda + (1-\delta_{{\rm B}_1})\delta_{{\rm S}_1}\lambda \\
&= \bigl[\delta_{{\rm B}_1} + (1-\delta_{{\rm B}_1})\delta_{{\rm S}_1}\bigr]\lambda 
= d\lambda.
\end{align*}
All terms in \(\varrho_1\) except \(\lambda_{\rm eff}\) are fixed system constants. Therefore, \(\varrho_1\) depends on the access probabilities only through \(\lambda_{\rm eff}\), which is uniquely determined by \(d\). Consequently, any two policies that yield the same value of \(d\) produce exactly the same conditional slot-success probability \(\varrho_1\) by Lemma~\ref{lemma:1}.
We consider the following two access policies with the same effective access probability \(d \in (0,1)\), thus 
\[ \text{ Policy A for block access is given by } \delta_{{\rm B}_1}=d,
\quad
\delta_{{\rm S}_1}=0. \]\[ 
\text{ Policy B for slotted access is given by }
\delta_{{\rm B}_1}=0,
\quad
\delta_{{\rm S}_1}=d. \]
Under Policy A, the controller accesses the entire block with probability \(d\). Hence,
$\pi_A = d\,\chi(\varrho_1) = d\,\varrho_1^2.$
Under Policy B, the controller accesses each slot independently with probability \(d\). Thus, each slot is successful with probability \(d\varrho_1\). Since the two slots are independent,
$\pi_B = \chi(d\varrho_1) = (d\varrho_1)^2 = d^2\varrho_1^2.$
Because \(0 < d < 1\), we have \(d^2 < d\), which implies \(\pi_B < \pi_A\).
Thus, two policies with the same per-slot transmission success probability can yield different block controllability probabilities. Therefore, maximizing the per-slot transmission success probability is not equivalent to maximizing block controllability.

\bibliographystyle{IEEEtran}      
\bibliography{references}
\end{document}